\documentclass[letterpaper, 10 pt, conference]{ieeeconf}
\IEEEoverridecommandlockouts                         
\usepackage{amsmath}
\usepackage{amssymb}
\usepackage{amsthm}
\newtheoremstyle{mystyle}%
  {\topsep}%
  {\topsep}%
  {\normalfont}%
  {}%
  {\bfseries}%
  {}%
  {.5em}%
  {}%
\theoremstyle{plain}
\newtheorem{thm}{Theorem}
\newtheorem{lem}{Lemma}

\newtheorem{assm}{Assumption}
\newtheorem{defn}{Definition}

\usepackage{mathtools}
\usepackage{siunitx}

\usepackage{algorithm}
\usepackage{algpseudocode}

\usepackage{lipsum}

\usepackage{subcaption}
\usepackage{subfig}

\usepackage[short,nocomma]{optidef}

\usepackage{tikz}

\newcommand\copyrighttext{%
  \footnotesize \textcopyright \the\year{} IEEE. Personal use of this material is permitted. Permission from IEEE must be obtained for all other uses, including reprinting/republishing this material for advertising or promotional purposes, collecting new collected works for resale or redistribution to servers or lists, or reuse of any copyrighted component of this work in other works.}

\newcommand\copyrightnotice{%
\begin{tikzpicture}[remember picture,overlay]
\node[anchor=south,yshift=10pt] at (current page.south) {\fbox{\parbox{\dimexpr0.75\textwidth-\fboxsep-\fboxrule\relax}{\copyrighttext}}};
\end{tikzpicture}%
}

\title{\LARGE \bf
Safe and Stable Closed-Loop Learning for\\Neural-Network-Supported Model Predictive Control
}

\author{Sebastian Hirt$^{1}$, Maik Pfefferkorn$^{1}$, and Rolf Findeisen$^{1}$%
\thanks{$^{1}$S. Hirt, M. Pfefferkorn and R. Findeisen are with the Control and Cyber-Physical Systems Laboratory,
        Technical University of Darmstadt, Germany
        {\tt\small \{sebastian.hirt, maik.pfefferkorn, rolf.findeisen\}@iat.tu-darmstadt.de}}%
}

\begin{document}
\maketitle
\copyrightnotice
\thispagestyle{empty}
\pagestyle{empty}

\begin{abstract}
Safe learning of control policies remains challenging, both in optimal control and reinforcement learning.
In this article, we consider safe learning of parametrized predictive controllers that operate with incomplete information about the underlying process.
To this end, we employ Bayesian optimization for learning the best parameters from closed-loop data.
Our method focuses on the system's overall long-term performance in closed-loop while keeping it safe and stable.
Specifically, we parametrize the stage cost function of an MPC using a feedforward neural network.
This allows for a high degree of flexibility, enabling the system to achieve a better closed-loop performance with respect to a superordinate measure.
However, this flexibility also necessitates safety measures, especially with respect to closed-loop stability.
To this end, we explicitly incorporated stability information in the Bayesian-optimization-based learning procedure, thereby achieving rigorous probabilistic safety guarantees.
The proposed approach is illustrated using a numeric example.
\end{abstract}

\section{Introduction}
Model predictive control (MPC) allows for optimal control of complex, nonlinear systems while considering constraints on the system's states and the control inputs \cite{findeisen2002introduction}. However, its reliance on accurate models and suitably chosen cost functions presents challenges. Traditionally, MPC employs iterative prediction and optimization based on a system model.
Nevertheless, obtaining accurate models can be difficult due to system complexity or limited knowledge.
Furthermore, choosing a suitable open-loop cost function and constraints to achieve the desired closed-loop performance while guaranteeing safety remains challenging.
Machine learning has emerged as a promising approach to address this issue by enabling data-driven model learning for MPC, e.g., \cite{hewing2020cautious,Pfefferkorn2022,maiworm2021online,nguyen2021robust, zieger2020towards}. Employing machine learning in control, however, introduces concerns regarding performance, stability and repeated feasibility \cite{Mesbah2022}. These challenges can be tackled by robust MPC formulations, which are, however, computationally expensive and lead to overly conservative control actions, while relying on uncertainty bounds that are challenging to obtain for machine-learning-based models. Additionally, a highly accurate prediction model does not always guarantee optimal closed-loop performance with respect to a superordinate performance measure, see \cite{kordabad2023reinforcement,gevers1993towards} and references therein.
Recent works explore hierarchical frameworks that address these limitations by combining optimization-based controllers with optimization and learning algorithms.
These frameworks often involve two layers: a high-level layer for long-term and global performance optimization, and a lower level predictive controller, typically with prediction horizon significantly shorter than task duration, responsible for short-term planning and control.
The former often exploits techniques from reinforcement learning, e.g., \cite{kordabad2023reinforcement,zanon2021safe}, or Bayesian optimization, e.g., \cite{paulson2023tutorial,piga2019performance, hirt2024learning}, and is in conjunction with the latter.
The hierarchical structure offers a promising approach to achieve desired closed-loop behavior while mitigating performance limitations arising from lack of knowledge about the process.

We consider MPC with short prediction horizon and relying on prediction models with significant model-plant mismatch due to lacking process insights.
We aim to achieve long-term performance optimization of the underlying nonlinear dynamical system through a multi-episode Bayesian-optimization-based learning approach. Specifically, Bayesian optimization \cite{garnett2023bayesian}, a sample-efficient optimization method, is employed to learn parameters of the MPC with the objective of optimizing a desired long-term performance metric.
To this end, a key aspect of our approach involves parametrizing the MPC stage cost function using an artificial neural network. This leverages the powerful representation capabilities of neural networks to handle a large class of functions. Essentially, parametrizing the cost function offers an alternative approach to model learning, motivated by the inherent connection between model-based predictions and the MPC cost function. 

Similar ideas have been explored in \cite{seel2022convex}, where convex neural networks are exploited for parametrizing the MPC and tuned via reinforcement learning.
However, such convex parametrizations restrict the full potential of using neural networks and it is challenging to establish robust safety guarantees in closed-loop both during learning and for the final result.
In \cite{hirt2023stability}, including stability information directly in the Bayesian optimization layer has been considered.
While the proposed approach is capable of providing Lyapunov-based stability certificates in closed-loop for the tuned controller, instability and constraint violations can occur during learning for a significant fraction of probed trajectories.
This restricts the application of the approach proposed in \cite{hirt2023stability} to non-safety-critical real-world-processes if the learning includes real-world data from experiments. 
We address these limitations by proposing a novel, stability-informed Bayesian optimization approach for controller tuning.
In particular, we exploit safe Bayesian optimization algorithms, such as used in \cite{berkenkamp2016safeb,wei2023highdimensional,krishnamoorthy2022safe}, to obtain probabilistic stability guarantees already during learning for the probed trajectories.

The remainder of this article is structured as follows.
We provide fundamentals on model predictive control, Gaussian process regression and safe Bayesian optimization in Section \ref{sec:fundamentals}.
Subsequently, the proposed safe learning approach is introduced in Section \ref{sec:main}, including the neural-network-based parameterization of the MPC scheme and the incorporation of stability constraints into the Bayesian optimization procedure.
In Section \ref{sec:simulation}, we underline the effectiveness and safety of the proposed approach in a simulation study before concluding in Section \ref{sec:conclusion}.

\section{Fundamentals}
\label{sec:fundamentals}
This section presents an overview of the control task and establishes the necessary theoretical foundations.
We outline the control objective, introduce parametrized model predictive control, and recapitulate on stability of dynamical systems.
Subsequently, we introduce Gaussian process regression for surrogate modeling before presenting an algorithm for safe Bayesian optimization.

\subsection{Problem Formulation}
We consider a nonlinear, discrete-time dynamical system
\begin{equation} 
\label{eqn:discrete_system_general}
    x_{k+1} = f(x_k, u_k),
\end{equation}
where $x_k \in \mathbb{R}^{n_\text{x}}$ are the system states, $u_k \in \mathbb{R}^{n_\text{u}}$ are the system inputs, $f: \mathbb{R}^{n_\text{x}} \times \mathbb{R}^{n_\text{u}} \to \mathbb{R}^{n_\text{x}}$ represents the (nonlinear) dynamics, and $k \in \mathbb{N}_0$ is the discrete time index.
Our control objective is to optimally steer system \eqref{eqn:discrete_system_general} from an initial state $x_0$ to a desired set-point $(x_d, u_d)$ while satisfying input and state constraints, and while maintaining stability.
To do so, we use model predictive control but consider the case of a significant model-plant-mismatch between the controller-intern prediction model and the real-world plant.
To address the model-plant mismatch, we parametrize the model predictive control formulation.
This approach incorporates additional degrees of freedom, enabling us to tune the MPC towards improved closed-loop performance.
Then, the objective is to find an appropriate parameterization that optimizes closed-loop performance while maintaining stability, even in the presence of the aforementioned mismatch.
To learn such a parametrization in a structured way, we exploit Bayesian optimization on closed-loop data.
Stability during the learning procedure is ensured by employing algorithms from the class of safe Bayesian optimization.
Note that the proposed approach is not limited to such regulation problems, but generalizes straightforwardly to a wide variety of scenarios, including imitation learning to meet expert demonstrations, or inference of suitable low-complexity prediction models if the system dynamics are too complex.

\subsection{Parametrized Model Predictive Control}
We consider an MPC formulation with $n_{\text{p}} \in \mathbb{N}$ parameters $\theta \in \Theta \subset \mathbb{R}^{n_{\text{p}}}$.
At every discrete time index $k$, given a set of parameters $\theta$, the MPC solves a parameterized optimal control problem given by
\begin{mini!}
    {\mathbf{\hat{u}}_k}{\left\{ J(x_k, \mathbf{\hat{u}}_k) = \sum_{i=0}^{N-1} l_\theta({\hat x}_{i \mid k}, {\hat u}_{i \mid k}) + E_\theta({\hat x}_{N \mid k}) \! \right\}\label{eqn:mpc_ocp_cost}}{\label{eqn:mpc_ocp}}{}
    \addConstraint{\forall i}{\in \{0, 1, \dots, N-1\}: \notag}{}
    \addConstraint{}{\hat x_{i+1\mid k} = \hat f_\theta(\hat x_{i\mid k}, \hat u_{i\mid k}), \ \hat x_{0 \mid k} = x_k,}{\label{eqn:mpc_ocp_model}}
    \addConstraint{}{\hat x_{i \mid k} \in \mathcal{X}_\theta, \ {\hat u}_{i \mid k} \in \mathcal{U}_\theta, \ \hat x_{N \mid k} \in \mathcal{E}_\theta.}{\label{eqn:mpc_ocp_constraints}}
\end{mini!}
Herein, $\hat{\cdot}_{i\mid k}$ denotes the model-based $i$-step ahead prediction at time index $k$, $\hat f_\theta(\cdot)$ is the (parametrized) prediction model and $x_k$ is a measurement of the system state at time index $k$.
The length of the prediction horizon is $N \in \mathbb{N}, N < \infty$ and $l_\theta(\cdot)$ and $E_\theta(\cdot)$ are the (parametrized) stage and terminal cost functions, respectively.
The constraints \eqref{eqn:mpc_ocp_constraints} are comprised of the (parametrized) state, input, and terminal sets $\mathcal{X}_\theta$, $\mathcal{U}_\theta$, and $\mathcal{E}_\theta$, respectively.
The solution of \eqref{eqn:mpc_ocp} results in the optimal input sequence $\mathbf{\hat{u}}_k^*(x_k; \theta)=[\hat u_{0 \mid k}^*(x_k; \theta),\dots,\hat u_{N-1 \mid k}^*(x_k; \theta)]$, of which the first element is applied to system \eqref{eqn:discrete_system_general}. 
Thus, repeatedly solving \eqref{eqn:mpc_ocp} at every time index $k$ defines the (parametrized) control policy $\hat u_{0 \mid k}^*(x_k; \theta)$.

\color{black}
\subsection{Stability of Dynamical Systems}
Stability is a crucial characteristic of safety in dynamical systems.
For a  proper stability definition, we use the following function classes.
\begin{defn}
(Comparison functions)
The function classes $\mathcal{K}$, $\mathcal{L}$ and $\mathcal{KL}$ are defined as
\begin{align}
    \mathcal{K} \coloneq \{& \kappa: \mathbb{R}_0^+ \rightarrow \mathbb{R}_0^+ \mid \kappa \in \mathcal{C}^0, \forall b > a: \kappa(b) > \kappa(a), \notag \\
    &\kappa(0) = 0 \} \notag \\
    \mathcal{L} \coloneq \{& \lambda: \mathbb{R}_0^+ \rightarrow \mathbb{R}_0^+ \mid \lambda \in \mathcal{C}^0, \forall b > a: \lambda(b) < \lambda(a), \notag \\
    &\lim_{t \rightarrow \infty} \lambda(t) = 0 \} \notag \\
    \mathcal{KL} \coloneq \{& \zeta: \mathbb{R}_0^+ \times \mathbb{R}_0^+ \rightarrow \mathbb{R}_0^+ \mid \zeta \in \mathcal{C}^0, \zeta(\cdot, t) \in \mathcal{K}, \notag \\
    & \zeta(r, \cdot) \in \mathcal{L} \}, \notag
\end{align}
where $\mathcal{C}^0$ is the class of continuous functions.
\end{defn}
Exploiting the above function classes, we define \textit{P-practically asymptotic stability}.
\begin{defn}
\label{defn:stability}
(P-practically asymptotic stability \cite{grune2017nonlinear}) Let $Y$ be a forward invariant set for system \eqref{eqn:discrete_system_general} and let $P \subset Y$ be a subset of $Y$, where $P$ is contained in a ball around the equilibrium, i.e., $P \subseteq \mathcal{B}_\nu(x_d), \nu > 0$ and $x_d \in P$. Then we say that a point $x_d \in Y$ is P-practically asymptotically stable on $Y$ if there exists $\zeta \in \mathcal{KL}$ such that
\begin{equation}
    \lVert x_k-x_d \rVert \leq \max \{ \zeta(\lVert x_0-x_d \rVert, k), \nu \}
\end{equation}
holds for all $x_0 \in Y$ and all $k \in \mathbb{N}_0$.
\end{defn}
Loosely speaking, $P$-practically asymptotic stability means that all system trajectories, defined by $x_k, k \in \mathbb{N}$, that start in $Y$ converge asymptotically to the ball $\mathcal{B}_\nu(x_d) := \{ x \in Y \mid \| x - x_d  \| \leq \nu \}$, containing the desired state $x_d$.
By $\| \cdot \|$, we denote some standard vector norm on $\mathbb{R}^{n_x}$.

\subsection{Gaussian Process Surrogate Models}\label{sec:gp}
For learning the parameters of the MPC \eqref{eqn:mpc_ocp} via Bayesian optimization from closed-loop data, a surrogate model of the mapping from the parameters $\theta$ to the closed-loop performance is required.
To infer such a model, we employ Gaussian process (GP) regression to closed-loop data of system \eqref{eqn:discrete_system_general} under MPC \eqref{eqn:mpc_ocp}.

Gaussian process regression allows learning of a probabilistic model of an unknown function $\varphi: \mathbb{R}^{n_\xi} \to \mathbb{R}, \xi \mapsto \varphi(\xi)$ from data.
Loosely speaking, a GP $g(\xi) \sim \mathcal{GP}(m(\xi), k(\xi, \xi^\prime))$ is a Gaussian probability distribution over an infinite-dimensional function space.
A GP is fully defined by its prior mean function $m: \mathbb{R}^{n_\xi} \to \mathbb{R}, \xi \mapsto \mathrm{E}[g(\xi)]$ and the prior covariance function $k: \mathbb{R}^{n_\xi} \times \mathbb{R}^{n_\xi} \to \mathbb{R}, (\xi, \xi') \mapsto \mathrm{Cov}[g(\xi), g(\xi')]$.

Our objective is to learn a meaningful prediction model of $\varphi(\cdot)$ that yields accurate estimates of unobserved function values $\varphi(\xi_*)$ at test inputs $\xi_*$.
To this end, we use a set of (noisy) function observations $\mathcal{D} = \{ (\xi_i, \gamma_i = \varphi(\xi_i) + \varepsilon_i) \mid i \in \{1, \dots, n_{\text{d}} \}, \varepsilon_i \sim \mathcal{N}(0, \sigma^2) \}$, where $\varepsilon_i$ is white Gaussian noise with variance $\sigma^2$.
For notational convenience, we collect all training inputs $\xi_i$ in the input matrix $\Xi \in \mathbb{R}^{n_{\text{d}} \times n_\xi}$ and the corresponding training targets in the vector $\gamma \in \mathbb{R}^{n_{\text{d}} \times 1}$ and write equivalently $\mathcal{D} = (\Xi, \gamma)$.
Via Bayesian inference, the information provided by the training data is then incorporated into the model.

This inference step yields the so-called posterior distribution $g(\xi_*) \mid \Xi, \gamma, \xi_* \sim \mathcal{N}(m^+(\xi_*), k(\xi_*, \xi_*))$ at $\xi_*$ with mean and variance given by
\color{black}
\begin{subequations}
\label{eqn:posterior_gp}
\begin{align}
    m^+(\xi_*) &= m(\xi_*)+k(\xi_*,  \Xi) k_{\gamma}^{-1} (\gamma-m(\Xi)), \label{eq:gp_postMean} \\
    k^+(\xi_*) &= k(\xi_*, \xi_*) - k(\xi_*, \Xi) k_{\gamma}^{-1} k(\Xi, \xi_*). \label{eq:gp_postVar}
\end{align}
\end{subequations}
Here, $k_{\gamma} = k(\Xi, \Xi) + \sigma^2 I$ and $I$ denotes the identity matrix.
The posterior mean \eqref{eq:gp_postMean} is an estimate for the unknown function value $\varphi(\xi_*)$, and the posterior variance \eqref{eq:gp_postVar} quantifies the uncertainty of the predicted value \cite{rasmussen2006gaussian}.

The prior mean and covariance function are design choices and often involve free hyperparameters.
The latter need to be carefully adjusted to the underlying problem in order to obtain a meaningful model.
One approach to obtaining suitable hyperparameters is through evidence maximization based on the available training data $\mathcal{D}$, see \cite{rasmussen2006gaussian}.

To define how well the learned GP model represents the true function, we review the concept of well-calibrated GPs in the following definition.
\begin{defn}
\label{defn:gp_calibration}
    (Well-calibrated Gaussian process \cite{krishnamoorthy2022safe}) A Gaussian process model with posterior mean $m^+(\cdot)$ and posterior variance $k^+(\cdot)$ that is used to approximate the true function $\varphi(\cdot)$ is well-calibrated if the inequality
    \begin{equation}
        \label{eqn:gp_calibration}
        |m^+(\xi)-\varphi(\xi)| \leq \beta(\delta) \sqrt{k^+(\xi)}
    \end{equation}
    holds for all $\xi \in \chi$ in a region of interest $\chi \subset \mathbb{R}^{n_\xi}$ with a probability of at least $1-\delta$ for any $\delta \in (0, 1)$.
\end{defn}

It is possible to employ a wide variety of results from literature to find confidence parameters $\beta(\delta)$ for given $\delta \in (0,1)$ such that a GP is well-calibrated.
In the following, we provide one such result.
\begin{lem}\label{lem:gp_calibration}(\cite{chowdhury2017kernelized})
    Let the covariance function $k$ be continuous and positive definite and assume that $\varphi \in \mathcal{H}_k$ is contained in the reproducing kernel Hilbert space (RKHS) $\mathcal{H}_k$ associated with $k$. Assume further that the RKHS norm of $\varphi$ is bounded from above by $B \in \mathbb{R}_+, B < \infty$, i.e., $\| \varphi \|_k \leq B$, and that the training data are subject to $R$-Gaussian measurement noise.
    Then, the GP is well-calibrated in the sense of Definition \ref{defn:gp_calibration} for $\beta(\delta) = B + R \sqrt{\gamma + 1 + \ln(\delta^{-1})}$, where $\gamma$ is the maximum information gain.
\end{lem}

Further similar results exist, such as \cite{Capone2022}, \cite{Fiedler2021} and references therein.
While many results rely on continuous, positive definite kernels, this is not a major restriction.
The reason is that the class of continuous, positive definite kernels includes many so-called universal kernels, which are capable of approximating any continuous function on a compact set with arbitrary accuracy \cite{steinwart2008support}.
Specifically, even if $\varphi \in \mathcal{C}^0$ is not contained in the RKHS of such a kernel, there exists $\hat{\varphi} \in \mathcal{H}_k$ such that $\| \varphi - \hat{\varphi} \|_\infty \leq \epsilon$ for given $\epsilon > 0$, which renders results such as Lemma \ref{lem:gp_calibration} applicable in such cases.

\subsection{Safe Bayesian Optimization}
Bayesian optimization (BO) serves as an iterative optimization strategy for problems involving expensive-to-evaluate black-box functions \cite{garnett2023bayesian}.
In this work, we leverage BO to maximize closed-loop performance of system \eqref{eqn:discrete_system_general} under MPC \eqref{eqn:mpc_ocp} with respect to the controller parameters $\theta$ while adhering to specified safety constraints.
Note that the functional relationships between the parameters $\theta$ and the closed-loop measures for performance as well as safety cannot be expressed in closed form and are costly to evaluate, motivating the use of BO.

More specifically, we employ BO to iteratively solve the optimization problem
\begin{align}
\begin{split}
    \label{eqn:bo_optimization_problem}
    \theta^* &= \arg \max_{\theta \in \Theta} \left\{G_0(\theta)\right\} \\
    & \text{s.t.} \ \forall i \in \{1, \ldots, n_{\text{bc}} \}: G_i(\theta) \geq 0,
\end{split}
\end{align}
where $G_0(\cdot)$ quantifies the closed-loop performance and $G_i(\cdot), i \in \{1, \ldots, n_{\text{bc}} \}$ are $n_{\text{bc}} \in \mathbb{N}$ black-box safety constraints on the closed-loop realizations.
To solve \eqref{eqn:bo_optimization_problem}, BO exploits probabilistic surrogate models of the unknown black-box functions $G_0(\cdot)$ and $G_i(\cdot), i \in \{1, \ldots, n_{\text{bc}} \}$.
Commonly, Gaussian process models are used because of their efficiency given a limited amount of training data.
During learning, the GP surrogate models are sequentially updated using data obtained from probing different MPC parameterizations, which induces the iterative nature of the BO procedure.
Particularly, in each iteration $n \in \mathbb{N}$, we
\begin{enumerate}
    \item[1)] select the next parameter set of interest $\theta_n$ and conduct a closed-loop run using MPC \eqref{eqn:mpc_ocp} to generate a new data point $\{ \theta_n, G_0(\theta_n), \ldots, G_{n_{\text{bc}}}(\theta_n) \}$,
    \item[2)] update the training data set with the newly observed data point: $\mathcal{D}_{n+1} \leftarrow \mathcal{D}_n \cup \{ \theta_n, G_0(\theta_n), \ldots, G_{n_{\text{bc}}}(\theta_n) \}$,
    \item[3)] update the (posterior) GP models based on $\mathcal{D}_{n+1}$, including optimization of the hyperparameters.
\end{enumerate}
Here, $\mathcal{D}_n$ denotes the training data set with data points observed up to iteration $n$.

To conduct the sequential learning procedure in a structured way, an acquisition function is used to guide the selection of new parameter sets $\theta_{n}$ towards the optimal parameters $\theta^*$ for increasing $n$.
An acquisition function $\alpha_0: \mathbb{R}^{n_{\text{p}}} \to \mathbb{R}, \theta \mapsto \alpha_0(\theta; \mathcal{D}_n)$ exploits the surrogate model of the unknown function $G_0(\cdot)$ to evaluate the utility of a set of parameters $\theta_n$ with respect to closed-loop performance based on the collected data up to iteration $n$.
Exploiting the uncertainty information of the surrogate model, the learning procedure is able to trade off exploration and exploitation in the parameter space.
The next set of parameters is then determined by
\begin{equation}
    \label{eqn:bo_update}
    \theta_{n+1} = \arg \max_{\theta \in \Theta} \left\{ \alpha_0(\theta; \mathcal{D}_n) - \tau \ \sum_{i=1}^{n_\text{bc}} \alpha_i(\theta; \mathcal{D}_n) \right\}.
\end{equation}
While $\alpha_0(\theta; \mathcal{D}_n)$ focuses on the closed-loop performance, the penalty parameter $\tau \in \mathbb{R}_+$ and the penalty terms $\alpha_i: \mathbb{R}^{n_{\text{p}}} \to \mathbb{R}, \theta \mapsto \alpha_i(\theta; \mathcal{D}_n)$ are used to include the black-box constraints $G_i(\cdot), i \in \{1, \dots, n_\text{bc} \}$.
The penalty terms are given as logarithmic barriers of the form
\begin{equation}
    \label{eqn:log_barrier_term}
    \alpha_i(\theta;\mathcal{D}_n) = \log \left(m_{G_i}^+(\theta) - \beta(\delta) \sqrt{k_{G_i}^+(\theta)} \right).
\end{equation}
Herein, $m_{G_i}^+(\cdot)$ and $k_{G_i}^+(\cdot)$ are the posterior mean and variance of the GP model associated with the constraint $G_i(\cdot)$.
More precisely, $\alpha_i(\cdot), i \in \{1, \dots, n_\text{bc} \}$ in \eqref{eqn:log_barrier_term} are used to incorporate the lower-confidence bound of the learned black-box constraints into the BO procedure \cite{krishnamoorthy2023tuning}.

To allow for safe learning, a safe initial data set $\mathcal{D}_0$ must be known.
This is formulated in the following standard assumption.
\begin{assm}
    \label{assm:initial_safe_set}
    There exists a safe initial data set $\mathcal{D}_0$ such that for all data points $(\theta, G_0(\theta), ..., G_{n_\mathrm{bc}}(\theta)) \in \mathcal{D}_0$, it holds that $G_i(\theta) \geq 0, \forall i \in \{ 1, \dots, n_\text{bc} \}$.
\end{assm}

In addition to that, the constraint functions must be exactly represented by their GP surrogate models according to the following assumption.
\begin{assm}
    \label{assm:well_calibrated_gp}
    The surrogate Gaussian process models of $G_i, i \in \{ 1, \dots, n_\text{bc} \},$ are well-calibrated in the sense of Definition \ref{defn:gp_calibration}.
\end{assm}

We now formulate the following result on parameter selection with guaranteed probabilistic constraint satisfaction.
\begin{lem}
    \label{lem:safe_parameter_learning}
    \textit{(Safe parameter learning \cite{krishnamoorthy2023tuning})} 
    Under Assumptions \ref{assm:initial_safe_set} and \ref{assm:well_calibrated_gp}, for some $\delta \in (0,1)$, and for any unconstrained acquisition function $\alpha_0(\theta; \mathcal{D}_n)$, $G_i(\theta_n) \geq 0$ holds for all $i \in \{ 1, \dots, n_\text{bc} \}$ and for $\theta_n$ chosen according to \eqref{eqn:bo_update} with probability at least $1-\delta$.
\end{lem}

As discussed in Section \ref{sec:gp}, satisfaction of Assumption \ref{assm:well_calibrated_gp} can be achieved using standard results if the functions $G_i(\theta), i \in \{ 1, \dots, n_\text{bc} \}$ are uniformly continuous, which is not restrictive.
However, it is worth noting that rigorous error bounds obtained according to Lemma \ref{lem:gp_calibration} and similar results are often overly conservative.

\section{Safe Learning of a Neural Cost Function}
\label{sec:main}
In this section, we outline our approach for achieving stable and safe closed-loop learning.
We introduce a neural stage cost function and detail a Bayesian optimization procedure to learn its parameters effectively.
Subsequently, we derive (probabilistic) stability and safety guarantees for the learning procedure.

\subsection{Neural Cost Function}
In practice, it is challenging to obtain accurate prediction models $\hat f(x, u)$ for the true system dynamics $f(x, u)$.
We circumvent the usual approach of learning an open-loop dynamics model from data.
Our approach adopts a more direct method by directly learning the MPC stage cost from closed-loop experiments.
This approach is motivated by the understanding that the model's dynamics are ultimately reflected in the MPC cost function, as expressed by the following relationship
\begin{align*}
    J(x_k, \mathbf{\hat{u}}_k) & = \sum_{i=0}^{N-1} l({\hat x}_{i \mid k}, {\hat u}_{i \mid k}) + E({\hat x}_{N \mid k}) \\
    & =  l(x_k, \hat{u}_{0 \mid k}) + \sum_{i=1}^{N-1} l \big( \hat{f}(\hat{x}_{i-1 \mid k}, \hat{u}_{i-1 \mid k}), \hat{u}_{i \mid k} \big) \\
    & \qquad \qquad \qquad \qquad \qquad + E \big( \hat{f}(\hat{x}_{N-1 \mid k}, \hat{u}_{N-1 \mid k}) \big).
\end{align*}
Although parametrizing only the cost function cannot entirely compensate for model errors, it is advantageous as the structure of the cost function is often simpler than that of the system dynamics.
This makes cost function learning computationally more attractive than learning open-loop prediction models, especially since the latter do not guarantee optimal closed-loop performance \cite{gevers1993towards}.
We use a quadratic nominal stage cost function for regulation to the desired set-point $(x_d, u_d)$,
\begin{equation}
    \label{eqn:quadratic_state_cost}
    l(x,u) = (x-x_d)^\top Q (x-x_d) + (u-u_d)^\top R (u-u_d),
\end{equation}
where $Q \in \mathbb{R}^{n_{\text{x}} \times n_{\text{x}}}, Q \succeq 0$ and $R \in \mathbb{R}^{n_{\text{u}} \times n_{\text{u}}}, R \succ 0$ are weighting matrices.
We formulate the parameterized stage cost of MPC \eqref{eqn:mpc_ocp} as
\begin{equation}
    \label{eqn:specific_stage_cost}
    l_\theta(x, u) = l(x,u) + l_{NN}(x; \theta).
\end{equation}
To ensure that $l_\theta(x_d, u_d) = 0$ holds, we further use
\begin{equation}\label{eq:NN_cost_term}
    l_{NN}(x; \theta) = y_{NN}(x)-y_{NN}(x_d),
\end{equation}
where $y_{NN}:\mathbb{R}^{n_{\text{x}}} \to \mathbb{R}$ is a feedforward neural network with $L$ layers, given by
\begin{equation*}
    y_{NN}(z) = W_L \sigma(W_{L-1} \sigma(...\sigma(W_1 z + b_1)... ) + b_{L-1}) + b_L.
\end{equation*}
Therein, $\sigma: \mathbb{R} \to \mathbb{R}$ is the activation function and $W_i$ and $b_i$, $i = 1, \ldots, L$, are the weight matrix and bias vector representing the transition from layer $i-1$ to layer $i$, where $i=0$ corresponds to the input layer.
We employ a quadratic terminal cost function $E(\hat{x}_{N \mid k}) = (\hat{x}_{N \mid k} - x_d)^\top P (\hat{x}_{N \mid k} - x_d)$, where $P \in \mathbb{R}^{n_x \times n_x}$ is obtained by solving the discrete-time Riccati equation underlying the associated linear-quadratic regulator problem for the linearized system dynamics around $x_d$.
We do not modify the terminal cost and use $E_\theta(\hat{x}_{N \mid k}) = E( \hat{x}_{N \mid k})$ in the following.

The set of parameters that are to be tuned by the BO algorithm is given by $\theta = \left\{ W_1, b_1, \dots, W_L, b_L \right\}$.
Thus, we consider a high-dimensional parameter space due to the large number of parameters within the neural network, necessitating efficient exploration during the learning process.
Consequently, a trade-off between the network's expressiveness, determined by its size, and computational complexity arises.
Note that the approach can be extended by additionally parameterizing the prediction model \eqref{eqn:mpc_ocp_model} or constraints \eqref{eqn:mpc_ocp_constraints}.
Since the changing parameterization of MPC \eqref{eqn:mpc_ocp} makes it challenging to derive a-priori guarantees on stability and state constraint satisfaction, these properties need to be enforced in the BO-based tuning procedure.

\subsection{Probabilistic Stability Guarantees During Learning}
It is challenging to restrict the BO-based parameter choices such that all closed-loop trials remain stable.
However, closed-loop stability is an important requirement for any control system.
To achieve safe parameter learning during the BO procedure, we incorporate stability requirements on the system, exploiting Definition \ref{defn:stability}.
We choose $\nu \in \mathbb{R}_+$ and $\zeta \in \mathcal{KL}$ such that
\begin{equation}
    \zeta(\lVert x_0 - x_k \rVert, k) = \rho \chi^k \lVert x_0-x_d \rVert_2
\end{equation}
for $\rho \in \mathbb{R}_+$ and $\chi \in (0,1)$.
Based thereon, we impose the stability condition
\begin{align}
    \label{eqn:stability_bo_constraint}
    G_1(\theta) = \min_{k \leq M} \left\{ \max \{ \rho \chi^k \lVert x_0-x_d \rVert_2, \nu \} - \lVert x_k-x_d \rVert_2 \right\},
\end{align}
where $M \in \mathbb{N}, M < \infty$ is the process run time.
This leads to the following result that probabilistically guarantees $P$-practically stable closed-loop trajectories during tuning.

\begin{thm}
    \label{thm:stable_bo}
    Select $\delta \in (0,1)$. Under Assumptions \ref{assm:initial_safe_set} and \ref{assm:well_calibrated_gp}, system \eqref{eqn:discrete_system_general} under MPC \eqref{eqn:mpc_ocp} with parameterization $\theta_n$ chosen according to \eqref{eqn:bo_update} satisfies $G_1(\theta_n) \geq 0$, and is thus $P$-practically stable, with probability at least $1-\delta$.
\end{thm}
\begin{proof}
    The result directly follows from Lemma \ref{lem:safe_parameter_learning} and exploiting Definition \ref{defn:stability}.
\end{proof}

Note that $\zeta$ and $\nu$ are design choices that can be used to enforce stability in the sense of Definition \ref{defn:stability}.
Clearly, $\zeta$ and $\nu$ need to be chosen such that the closed-loop trajectories resulting from the safe data set satisfy \eqref{eqn:stability_bo_constraint}.
Moreover, a stringent design of \eqref{eqn:stability_bo_constraint} results in less exploration, which negatively affects the performance improvement, while less stringent designs allow for faster performance improvements at the cost of more frequently unstable runs.
This further correlates with the choice of $\delta$, which is an upper bound on the fraction of closed-loop runs that violate \eqref{eqn:stability_bo_constraint}.
In contrast to traditional stability analysis for MPC, we do not rely on Lyapunov functions but enforce the stability condition from Definition \ref{defn:stability} directly.
Note that recursive feasibility of the MPC problem can, e.g., be achieved using robust approaches \cite{zanon2021safe}. As feasibility is independent of the proposed cost function parametrization, we refrain from a detailed discussion here.
While input constraints in \eqref{eqn:mpc_ocp_constraints} are naturally satisfied by the MPC, irrespective of its parameterization, this is not true for state constraints.
Thus, if state constraints are present, they need to be treated analogously and additionally imposed in the BO layer.

\begin{algorithm}
    \caption{Safe MPC Cost Learning via BO}
    \label{alg:safe_bo}
    \begin{algorithmic}[1]
    \Require Parametrized MPC policy $u_{0 \mid k}^*(x; \theta)$, parameter domain $\Theta$, initially safe data set $\mathcal{D}_0$, prior GP models of $G_0$ and $G_1$
    \For{$n = 0, 1, \cdots$}
        \State Determine posterior GPs from $\mathcal{D}_n$
        \State Maximize constrained acquisition function for $\theta_{n+1}$
        \State Run closed-loop experiment using $u_{0 \mid k}^*(x; \theta_{n+1})$
        \State Determine $G_0(\theta_{n+1})$ and $G_1(\theta_{n+1})$
        \State Update data set
    \EndFor
    \end{algorithmic}
\end{algorithm}

\section{Simulation Study}
\label{sec:simulation}
We illustrate the effectiveness and safety capabilities of the proposed approach in simulation.
After an introduction of our set-up, we show simulation results for cost function learning while imposing stability-informed constraints.

\subsection{Set-Up}
We employ a double pendulum as an example system. Its time-continuous dynamics are given by
\begin{align}\label{eq:pendulum_continuous}
\begin{split}
    \ddot \psi_1 =& (m_2  l_1 \dot{\psi}_1^2 s_{21} c_{21} + m_2 g s_2 c_{21} + m_2  l_2 \dot{\psi}_2^2 s_{21} \\
    &- (m_1 \! + \! m_2) g s_1) / (m_1 \! + \! m_2)  l_1 - m_2  l_1 c_{21}^2 + u \\
    \ddot \psi_2 =& (-m_2  l_2 \dot{\psi}_2^2 s_{21} c_{21} \! + \! (m_1 \! + \! m_2) (g s_1 c_{21} \! - \! l_1 \dot{\psi}_1^2 s_{21} \\
    &- g s_2)) / ( l_2 /  l_1) (m_1 \! + \! m_2)  l_1 - m_2  l_1 c_{21}^2,
\end{split}
\end{align}
where $s_i = \sin(\psi_i)$, $c_i = \cos(\psi_i)$ for $i \in \{1,2\}$ and $s_{21} = \sin(\psi_2 - \psi_1)$, $c_{21} = \cos(\psi_2 - \psi_1)$.
The system states $x = (\psi_1, \psi_2, \dot \psi_1, \dot \psi_2)$ are the angles $\psi_1, \psi_2$ of the pendulums links and the angular velocities $\dot{\psi}_1, \dot{\psi}_2$.
The control input $u \in [-50, 50]$ is applied as an acceleration acting on the base of the first link.
The parameters $m_1 = \SI{1}{\kg}, m_2 = \SI{1}{\kg}$ and $l_1 = \SI{1}{\m}, l_2 = \SI{1}{\m}$ are the masses and lengths of the links, respectively.
We discretize the above dynamics using a one-step fourth-order Runge-Kutta method with a sampling time of $T_s = \SI{0.05}{\s}$.
The control objective is to bring the pendulum from its initial state $x_0$ into the upright position defined by $x_d = (\pi, \pi, 0, 0), u_d = 0$ and to stabilize it there.

For a closed-loop trajectory of length $M$ under the MPC control law, we measure the closed-loop performance by the total weighted deviation of inputs and states from the set-point $(x_d, u_d)$.
Specifically, it is given by
\color{black}
\begin{equation}
\label{eqn:closed_loop_performance}
\begin{split}
    G_0(\theta) \! = \! \sum_{k=0}^{M} \lVert x_k \! - \! x_d \rVert_V^2 \! + \! \lVert u_k \! - \! u_d \Vert_W^2 \! + \! \lVert x_M \! - \! x_d \rVert_Z^2,
\end{split}
\end{equation}
where the weighting matrices $V, Z \in \mathbb{R}^{n_\text{x} \times n_\text{x}}, V \succeq 0, Z \succeq 0$ penalize the closed-loop states and $W \in \mathbb{R}^{n_\text{u} \times n_\text{u}}, W \succeq 0$ penalizes the MPC closed-loop control input.
While the structure of the above closed-loop cost is similar to the MPC cost function in \eqref{eqn:mpc_ocp_cost}, \eqref{eqn:closed_loop_performance} is based on an entire closed-loop trajectory of length $M$.
Thus, \eqref{eqn:closed_loop_performance} is global and long-term oriented, while the MPC cost function considers only the local cost over its prediction horizon of length $N \ll M$ and is thus more short-term oriented.
In our simulation study, we use CasADi to solve the MPC problem and BoTorch for safe Bayesian optimization.

\color{black}
\subsection{Safe Stability-informed Cost Function Learning}
We consider learning of the neural network parameterized MPC stage cost function while accounting for safe learning with respect to the stability condition encoded in \eqref{eqn:stability_bo_constraint}.
We use the discrete-time version of \eqref{eq:pendulum_continuous} with parameter estimates $\hat{l}_1 = \hat{l}_2 = \SI{1.2}{\m}$, $\hat{m}_1 = \SI{2}{\kg}$ and $\hat{m}_2 = \SI{0.5}{\kg}$ as a prediction model in the MPC, resulting in a significant model-plant mismatch.
The stage cost is parameterized by a neural network with one hidden layer comprised of seven neurons, resulting in $n_p = 43$ tunable parameters that constitute the weight and bias matrices.
We employ the hyperbolic tangent activation function in the hidden layer and a linear activation in the output layer.

For both the performance measure $G_0(\theta)$ and the stability constraint $G_1(\theta)$, we utilize zero-mean Gaussian process priors with Matérn covariance functions.
We derive the hyperparameters via evidence maximization and the posterior GPs according to \eqref{eqn:posterior_gp}.
We use an initial data set $\mathcal{D}_0$ comprised of $100$ randomly sampled safe parameterizations, which fulfill the stability constraint $G_1$. We obtain the corresponding performance values according to \eqref{eqn:closed_loop_performance} and levels of constraint satisfaction according to \eqref{eqn:stability_bo_constraint}.
Subsequently, the BO procedure is run for $400$ iterations, utilizing the expected improvement acquisition function.

To illustrate the probabilistic safety capabilities of our algorithm, the constraints are imposed using confidence scaling parameters of $\beta = 0.5$ and $\beta = 2$.
For simplicity, we correspondingly assume $\delta = 0.617$ and $\delta = 0.046$ according to the confidence levels of Gaussian distributions.
Additionally, we limit the amount of exploration during the BO procedure by increasing the parameter bounds $\Theta$ incrementally from run to run.
This is a standard practice and employed due to the fact that BO tends to over-explore in high-dimensional parameter spaces \cite{wei2023highdimensional}.

All closed-loop trajectories that were probed by Algorithm \ref{alg:safe_bo} are shown in Figure \ref{fig:sim_states} by gray lines for the states $\psi_1$ and $\psi_2$ of the pendulum and for $\beta = 2$.
The initially safe trajectory without a neural network cost (orange line) and the optimal learned trajectory (blue line) are highlighted.
Both the initial and the learned trajectory steer the system into the desired set-point (black dashed line).
However, the learned trajectory reaches the set-point faster and shows fewer oscillations and overshooting compared to the initial trajectory.
Furthermore, the learned controller satisfies the imposed stability condition \eqref{eqn:stability_bo_constraint} during most of the closed-loop runs.
This is in line with the probabilistic stability guarantees from Theorem \ref{thm:stable_bo}: constraint violations occur (i) in 26 runs for $\beta = 0.5$ ($\delta_\mathrm{empiric} = 0.065$), and (ii) in 3 runs for $\beta = 2$ ($\delta_\mathrm{empiric} = 0.005$).
In Figure \ref{fig:sim_comparison_functions}, we show the state norm over time for all closed-loop runs and the desired upper bound from the stability condition.
For both cases, the number of observed constraint violations is significantly lower than guaranteed by Theorem \ref{thm:stable_bo} despite the simplified choice of the confidence parameter.
This is because the BO algorithm not only probes parameterizations on the boundary of the safe parameter region (exploration) but also from its interior (exploitation).
Additionally, the safety level is increase by the algorithmic restriction of exploration, as discussed above.

Our simulation study shows that the neural network cost function modification can (partially) compensate for modeling errors while maintaining (probabilistic) safety during the BO-based learning procedure.

\begin{figure}[t]
    \centering
    \begin{subfigure}{0.5\textwidth}
        \includegraphics{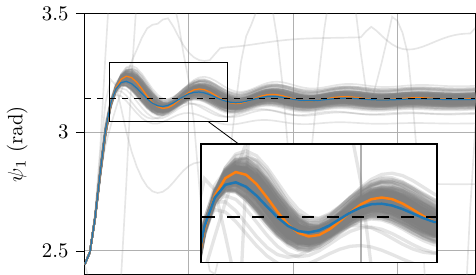}
    \end{subfigure}
    \begin{subfigure}{0.5\textwidth}
    \vspace{1mm}
        \includegraphics{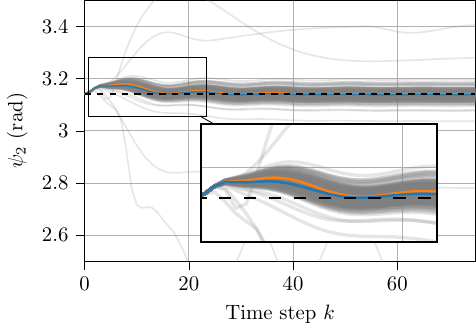}
    \end{subfigure}
    \vspace{-3mm}
    \caption{States $\psi_1$ (top) and $\psi_2$ (bottom) during the learning procedure for $\beta = 2$. We show the learned result (blue), the initial safe closed-loop run without a neural network stage cost (orange), and all intermediate closed-loop samples resulting from the BO procedure (gray).}
    \label{fig:sim_states}
\end{figure}

\begin{figure}[t]
    \centering
    \begin{subfigure}{0.5\textwidth}
        \includegraphics{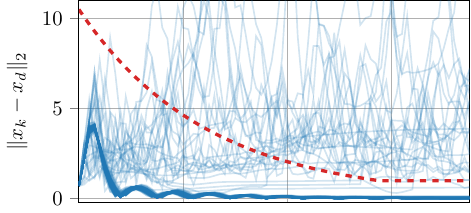}
    \end{subfigure}
    \begin{subfigure}{0.5\textwidth}
        \includegraphics{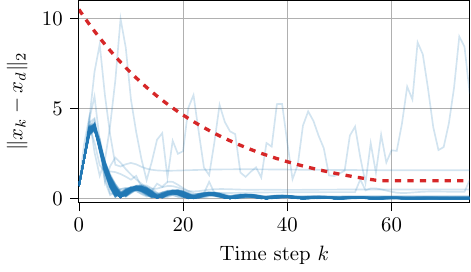}
    \end{subfigure}
    \vspace{-3mm}
    \caption{State norms (blue) for all closed-loop runs and constraint resulting from the stability condition $\max \{ \zeta(\lVert x_0-x_d \rVert, k), \nu \}$ (red). Results are shown for confidence parameters $\beta = 0.5$ (top) and $\beta = 2$ (bottom).}
    \label{fig:sim_comparison_functions}
\end{figure}

\section{Conclusions}
\label{sec:conclusion}
We presented a hierarchical learning and control scheme, consisting of a low-level parametrized model predictive controller and a higher-level Bayesian optimization layer for automated and safe parameter learning.
The cost function was parameterized by a feedforward neural network, enabling a flexible framework to guide the MPC controller towards achieving a desired closed-loop behavior.
This learned cost function can incorporate higher-level performance measures and (partially) compensate for model uncertainties commonly encountered in practice.
To ensure stability during learning, we introduced stability constraints on the Bayesian optimization procedure.
Using an initial set of safe parameters, we were able to extend the safe parameter region towards an improved closed-loop behavior while maintaining (probabilistic) stability.
We demonstrated the effectiveness of our approach in a simulation study, showcasing its ability to achieve both performance and safety objectives during learning, even in scenarios with significant model-plant mismatch.

Future research will be dedicated towards extending the proposed approach to richer parameterizations of the MPC, inducing even higher-dimensional parameter spaces.
This may involve hybrid approaches that combine reinforcement learning and Bayesian optimization to capitalize on the strengths of both techniques.

\bibliographystyle{ieeetr}
\bibliography{bibliography}

\end{document}